\documentclass[submission,copyright,creativecommons]{eptcs}
\providecommand{\event}{QPL 2020} 
\usepackage{underscore}           
\usepackage[utf8]{inputenc}
\usepackage[english]{babel}
\usepackage{amsmath}
\usepackage{amsfonts}
\usepackage{amssymb}
\usepackage{amsthm}
\usepackage{mathtools}
\usepackage{enumerate}
\usepackage{graphicx}
\usepackage{subcaption}

\newtheorem{theorem}{Theorem}
\newtheorem{corollary}{Corollary}
\newtheorem{lemma}{Lemma}
\theoremstyle{definition}
\newtheorem{definition}{Definition}

\DeclareMathOperator{\Pa}{Pa}
\DeclareMathOperator{\An}{An}

\title{Equivalence of Grandfather and Information Antinomy\\Under Intervention}
\author{\"Amin Baumeler
	\institute{Institute for Quantum Optics and Quantum Information (IQOQI), Boltzmanngasse 3, 1090 Vienna, Austria}
	\institute{Faculty of Physics, University of Vienna, Boltzmanngasse 5, 1090 Vienna, Austria}
	\institute{Facolt\`a indipendente di Gandria, Lunga scala, 6978 Gandria, Switzerland}
	\email{aemin.baumeler@oeaw.ac.at}
\and
	Eleftherios Tselentis
	\institute{Institute for Quantum Optics and Quantum Information (IQOQI), Boltzmanngasse 3 1090 Vienna, Austria}
	\email{eleftheriosermis.tselentis@oeaw.ac.at}
}
\def\titlerunning{Equivalence of Grandfather and Information Antinomy}
\def\authorrunning{\"A. Baumeler \& E. Tselentis}

\begin{document}
\maketitle
\begin{abstract}
	\noindent
	Causal loops, {\it e.g.}, in time travel, come with two main problems.
	Most prominently, the grandfather antinomy describes the potentiality to inconsistencies: a problem of logical nature.
	The other problem is called information antinomy and is lesser known.
	Yet, it describes a variant of the former: There are not {\em too few\/} consistent solutions---namely none---but {\em too many}.
	At a first glance, the information antinomy does not seem as problematic as the grandfather antinomy, because there is no apparent logical contradiction.
	In this work we show that, however, both problems are equivalent under interventions: If parties can intervene in such a way that the information antinomy arises, then they can also intervene to generate a contradiction, and {\em vice versa}.
\end{abstract}

\section{Historical background, motivation, and result}
\label{sec:intro}
Causal loops are loops in cause-effect relations such that some event~$Q$ not only is an effect of another event~$P$---its cause---, but also is the cause of~$P$.
The discussion on causal loops entered the realm of physics more than a century ago.
Einstein~\cite{EinsteinCTC}, while developing the theory of general relativity, expressed his doubt that in general relativity time travel might be possible: Every world line in special relativity is not closed, the same cannot be said about general relativity.
After Einstein asked Carath\'eodory~\cite{EinsteinToCarathedory1,EinsteinToCarathedory2} to resolve this question, Lanczos~\cite{Lanczos:1924kn}, others (see, {\it e.g.}, Ref.~\cite{Stockum:1937}), and most notably G\"odel~\cite{Godel1949} found solutions to the equations of general relativity that describe causal loops.
In G\"odel's words:~``[I]f~$P,Q$ are any two points on a world line of matter, and~$P$ precedes~$Q$ on this line, there exists a time-like line connecting~$P$ and~$Q$ on which~$Q$ precedes~$P$; {\it i.e.,} it is theoretically possible in these worlds to travel into the past, or otherwise influence the past''~\cite{Godel1949}.
In the '90s, researchers around Thorne and Novikov started to investigate such causal loops and asked whether causal loops might lead to inconsistencies.
{\it E.g.}, is it possible that a time travelling billiard ball kicks its younger self off course in such a way that the younger self does {\em not\/} time travel?~\cite{Friedman1990,Echeverria1991}
Towards answering this question, Novikov formulated the self-consistency principle~\cite{Novikov1989,Friedman1990} which states that only self-consistent solutions to the dynamics on a causal loop occur, and that locally, physics is kept unchanged.
This means that the physical laws must be invariant under the absence or presence of a causal loops.\footnote{Note that the self-consistency principle without this addendum is trivial: In case of inconsistent dynamics, we simply change the description of the physical world, {\it e.g.}, by allowing for parallel universes, such that every inconsistent solution becomes a consistent one.}
At the same time, by approaching causal loops from a circuit-diagrammatic point of view, Deutsch~\cite{Deutsch1991} argued that with the help of quantum theory inconsistencies can be overcome.
A later quantum model for time travel based on circuits overcomes this issue as well, yet in a different way (see, {\it e.g.}, Refs.~\cite{Bennett,Pegg:2001wa,Svetlichny:2009ve,Svetlichny:2011gq,Lloyd2011,Allen2014}).

Within the last decade, causal loops became a topic of research again.
During the studies of higher-order quantum maps, {\it i.e.}, quantum maps of quantum maps, and among other frameworks~\cite{Chiribella2008,Pollock2018}, the process-matrix framework was developed~\cite{Oreshkov2012}.
A key feature of that latter framework is that it allows for correlations among distant parties that cannot be simulated causally.
Causal inequalities~\cite{Baumeler3parties,simplestcausalinequalities,Abbott2016,ClassicalNC1} limit the space of possible correlations where the parties cannot communicate through causal loops.
That framework, however, leads to violations of such Bell-like causal inequalities.
Moreover, the classical special case of that quantum framework violates causal inequalities as well~\cite{Baumeler2016}; the framework allows for causal loops~\cite{phd,barrett2020cyclic}.

Most of preceding work tries to exclude inconsistencies from causal loops.
This problem of inconsistencies is famously known as the {\em grandfather antinomy}.
The story to illustrate that problem is the following, where we divert from the usual homicide plot to a technicide plot.
Imagine a robot is programmed in such a way to travel to the past to encounter its younger self.
Once the robot meets its younger self, it disassembles it.
So, if the robot time travels, the robot does not time travel.
But now, since the robot does not time travel, it will time travel, {\it etc.}: a logical contradiction.
A simple instance of this problem is obtained with a NOT gate, where a bit is flipped and then looped back (see Fig.~\ref{fig:NOT}).
Note that throughout this article, a loop is {\em not\/} a feedback loop where a map is repeatedly applied, one application after the other, but of logical nature instead (the loop introduces constraints between the input and the output).
\begin{figure}[t]
	\centering
	\begin{subfigure}[t]{.4\textwidth}
		\centering
		\includegraphics[scale=0.3]{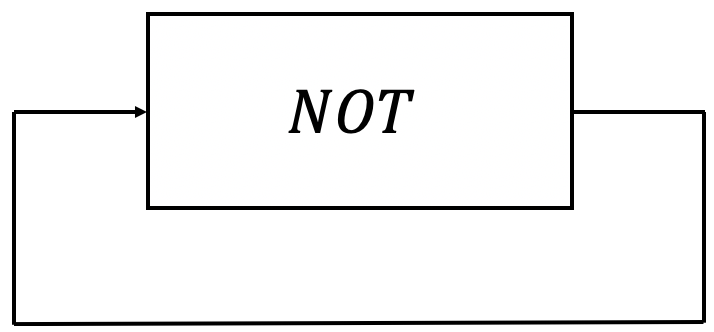}
		\caption{Grandfather antinomy: If, before the NOT gate, the bit takes value~$a$, then it takes value~$\neg a$ after. Then again,~$\neg a$ is looped back, which means that {\em before\/} the NOT gate, it takes value~$\neg a$: a logical contradiction.}
		\label{fig:NOT}
	\end{subfigure}
	\qquad
	\begin{subfigure}[t]{.4\textwidth}
		\centering
		\includegraphics[scale=0.3]{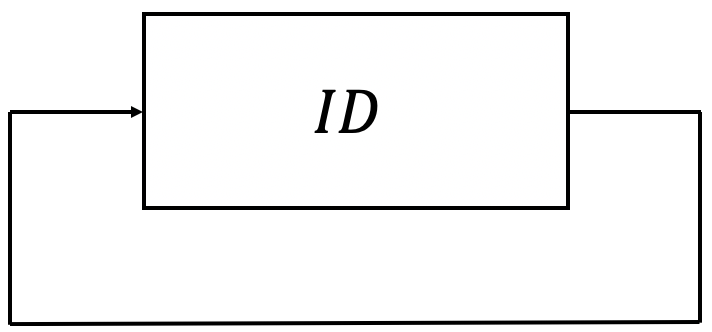}
		\caption{Information antinomy: If, before the identity gate, the bit takes value~$a$, then it takes value~$a$ after. Both values~$a=0$ as well as~$a=1$ are consistent solutions. Yet, what is the value of~$a$?}
		\label{fig:ID}
	\end{subfigure}
	\caption{Schematic representations of the grandfather and of the information antinomy.}
\end{figure}

Another problem, albeit lesser known, arises if too many consistent solutions occur.
This problem is known as the {\em information antinomy\/} and also carries different names\footnote{Already the multitude of names given to that problem suggests its inferiority when compared to the grandfather antinomy.} such as bootstrapping paradox, uniqueness ambiguity~\cite{Allen2014,Baumeler2018}, or ontological paradox~\cite{wuethrich}.
This problem is often illustrated with the following story.
Imagine a person wakes up one morning and finds, next to the bed, a book that contains a proof of a longstanding mathematical problem.
Later, this person travels to the past and places the book she or he found next to her or his bed.
If we analyse this story, we see that the book is given to that person by her/himself.
Yet, this story is problematic for two reasons.
Firstly, where does this proof come from?
We have complex information that arises out of nowhere.\footnote{Deutsch considers this problem as more severe compared to the grandfather antinomy and rejects it as {\em creationism\/}~\cite{Deutsch1991}.}
Secondly, why is the proof written in that way as it is written?
It could have been written in any other language, or, more drastically, the book might contain a proof for another longstanding mathematical problem.
Both problems we just discussed, however, are the same.
Given the boundary conditions multiple consistent outcomes might occur (see Fig.~\ref{fig:ID}).
A theory suffering from the information antinomy thus fails to provide predictions, even probabilistically~\cite{Allen2014}.

However, one might consider the information antinomy as less severe when compared to the grandfather antinomy.
The reason for such a judgement is often that the grandfather antinomy embodies a {\em logical\/} contradiction, while the information antinomy seems to be unproblematic from a logical point of view.
In this work we put light on this dichotomy and show that both problems are equivalent: They form two sides of the same coin.
This equivalence holds if the dynamics allow for intervention, {\it i.e.}, parties that are free to apply local transformations.
We show this equivalence in the deterministic setting without referring to probability theory or quantum theory.
The core in showing this equivalence is a theorem that is complementary to what is shown in Ref.~\cite{Baumeler2019}.
In Ref.~\cite{Baumeler2019} it is shown that, if for any choice of interventions parties can make, consistent solutions always exist, then the solution to the dynamics is {\em unique}.
Phrased differently, that result states that if no grandfather antinomy arises, then only a single consistent solution to the dynamics exists.
From this follows that the information antinomy does not arise either.
In this article we show that, if for any choice of interventions parties can make, no information antinomy arises ({\it i.e.}, not more than one consistent solution exists), then, just as in the previous case, the solution to the dynamics is {\em unique}.
From this, again, it follows that the absence of the information antinomy implies the absence of the grandfather antinomy.

Before we present the outline of this article, we reflect on how the information antinomy arises in previous articles.
In the mathematical experiments where billiard balls are thrown into a time machine, consistent dynamics were always found.
Yet, surprisingly, the authors discovered that ``dangerous'' boundary conditions lead to an infinity of consistent dynamics~\cite{Echeverria1991}.
In Deutsch's model~\cite{Deutsch1991}, then again, the information antinomy is {\em mitigated\/} by defining that the unique solution is the uniform mixture of all consistent solutions.
In contrast, the process-matrix framework~\cite{Oreshkov2012} seems not to suffer from this antinomy.
That is the case at least in the classical special case thereof~\cite{Baumeler2016}: The grandfather as well as the information antinomy never arises~\cite{Baumeler2016fixed}.

In the next section we describe causal models and provide the necessary definitions.
After that we show the above stated uniqueness results, from which the equivalence follows.
Finally, we conclude.

\section{Causal structure and interventions}
The reader interested on causal models is referred to the book by Pearl~\cite{Pearl} and to the recent articles, {\it e.g.,} to Refs.~\cite{Allen2016,barrett2019quantum,barrett2020cyclic}.
We need some notation before we can define the relevant mathematical objects.
Let~\mbox{$G=(V,E)$} be a directed graph, where~$V\subseteq\mathbb N$ denotes the set of vertices and where~$E$ is a relation describing the edges:~$(u,v)\in E$ if and only if there is an edge from~$u$ to~$v$.
For a vertex~$v$, the set~\mbox{$\Pa(v):=\{u\in V\mid (u,v)\in E\}$} is the set of all parents of~$v$.
A vertex~$u$ is called an ancestor of~$v$ if there exists a directed path from~$u$ to~$v$.
The set of ancestors of~$v$ is defined as~$\An(v)$.
Since we are dealing with deterministic dynamics, we define a causal structure as follows:
\begin{definition}[Causal structure]
	A {\em causal structure\/} is a tuple~$(G,X,\mu)$ where
		$G=(V,E)$ is a directed graph,
		$X$ is a family of sets~$\{X_v \mid |X_v| \geq 2\}_{v\in V}$, and 
		$\mu$ is a family of functions~$\{\mu_v: \bigtimes_{u\in\Pa(v)}X_u \rightarrow X_v\}_{v\in V}$.
\end{definition}
The parents of a vertex~$v$ are the {\em causes\/} of the {\em effect\/}~$v$.
The value every vertex takes is computed from a function of all its parents.
The condition that every vertex can take at least two values, {\it i.e.}, the condition~$|X_v| \geq 2$ for all~$v\in V$, is natural: A vertex~$v$ that can take only one value cannot be considered a cause nor an effect.
For a vertex~$v$ with zero in degree,~$\mu_v$ is a trivial function, {\it i.e.}, a constant.
Note that the sets associated to the vertexes do not have to be finite or discrete; they could, {\it e.g.}, contain all real numbers.

To allow for interventions, and similarly to the classical interventional model of Ref.~\cite{Allen2016} or the classical split nodes of Ref.~\cite{barrett2019quantum}, we augment a causal structure with parties.\footnote{We divert from the definitions in the mentioned articles in order to be more general; for split nodes the ``input'' and ``output'' of a party are elements from the same set. Any causal structure with split nodes can be transformed into a causal structure as we define it here.}
The main idea is that every party is a vertex~$i$ and can freely choose\footnote{By this local physics does not depend on whether the causal structure is cyclic or acyclic (see Section~\ref{sec:intro}). What we mean by ``freely choose'' is that every intervention is possible.} the function~$\mu_i$.
Furthermore, if we remove all incoming edges to the parties, then the graph has no directed cycles.\footnote{The reason for this condition is that otherwise interventions might have no influence on whether the grandfather or information antinomy arises. {\it E.g.}, imagine a causal structure with a detached loop free of any party such that the loop will always lead to an inconsistency. Our result can be read in the following way: If the parties {\em can\/} generate the grandfather antinomy, then they {\em can\/} also generate the information antinomy and vice versa.}
\begin{definition}[Causal structure with interventions and induced function]
	A {\em causal structure with interventions\/} is a tuple~$(G,X,\mu\setminus \{\mu_v\}_{v\in P},P)$ where
		$(G,X,\mu)$ is a causal structure,
		$P\subseteq V$ is a non-empty set of {\em parties}, and
		$(V,E\setminus \{(u,v)\in E\mid v\in P\})$ is a directed {\em acyclic\/} graph.
	For every party~$i\in P$ we define
		the {\em input space\/} as~$\mathcal{I}_i:=\bigtimes_{j\in\Pa(i)}X_j$, and the {\em output space\/} as~$\mathcal{O}_i:=X_i$,
		the {\em intervention\/} as a function~$f_i:\mathcal{I}_i\rightarrow\mathcal{O}_i$ of his or her choice, and
		the {\em induced function\/} as~$\omega_i:\bigtimes_{j\in P} \mathcal O_j\rightarrow \mathcal I_i$, as the function described by the causal structure.
		The {\em $|P|$-party induced function\/}~$\omega$ is defined as the list~\mbox{$\omega=(\omega_i)_{i\in P}$} and has signature~$\bigtimes_{i\in P}\mathcal O_i\rightarrow \bigtimes_{i\in P}\mathcal{I}_i$.
	Upon intervention~$\{f_i\}_{i\in P}$, {\em consistent assignment\/} of values to the vertexes exists if and only if there exists a family~$\{x_v\}_{v\in V}$ such that for every party~$i$ the value~$x_i$ equals~$f_i$ applied to the corresponding values, and for every vertex~$v\in V\setminus P$ the value~$x_v$ is equal to~$\mu_x$ applied to the corresponding values.
\end{definition}
An example of such a causal structure is shown in Fig.~\ref{fig:causalstructure}.
\begin{figure}[t]
	\centering
	\begin{subfigure}[t]{.4\textwidth}
		\centering
		\includegraphics[width=.8\textwidth]{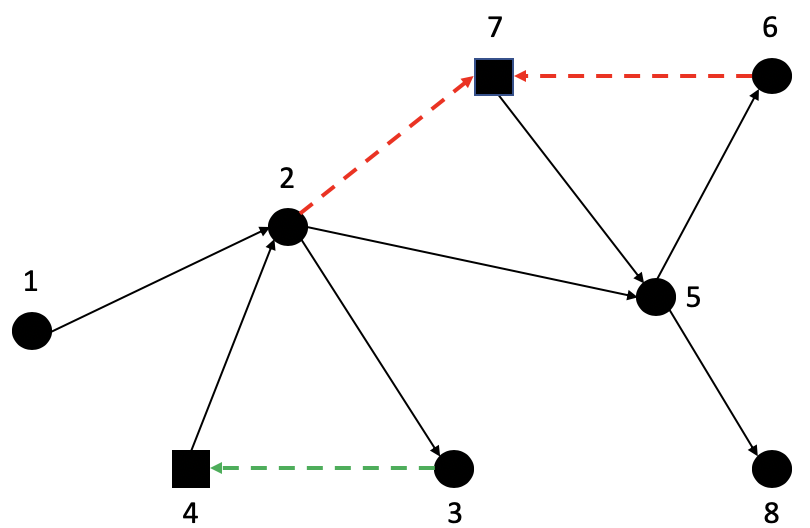}
		\caption{An example of a causal structure with interventions. Here, the vertexes~$4$ and~$7$ represent the parties. Thus, party~$4$ can intervene on the function along the edge from~$3$ to~$4$, and party~$7$ can intervene by specifying the function~$f_7:X_2\times X_6\rightarrow X_7$.}
		\label{fig:causalstructure}
	\end{subfigure}
	\qquad
	\begin{subfigure}[t]{.4\textwidth}
		\centering
		\includegraphics[width=\textwidth]{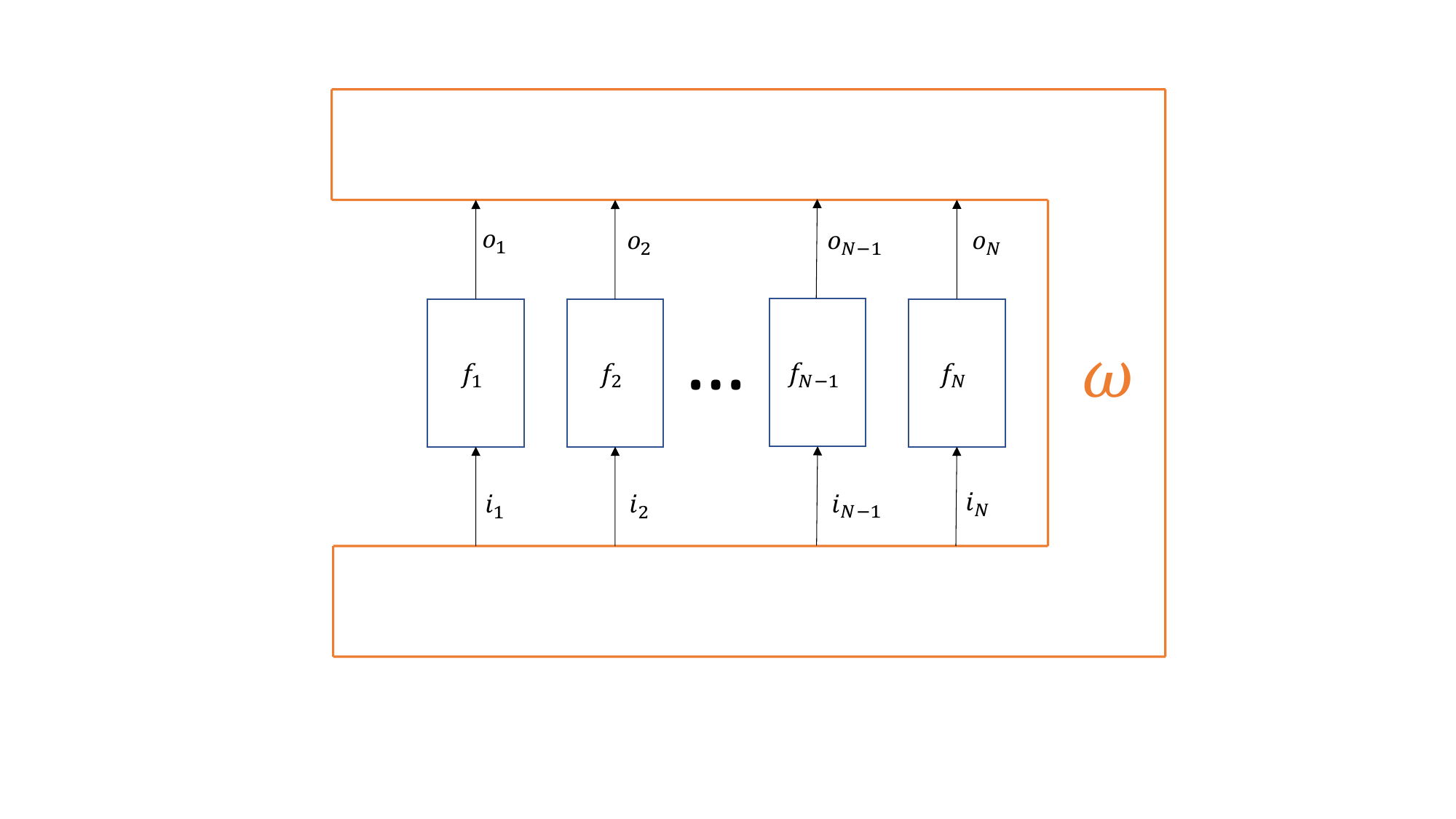}
		\caption{A pictorial depiction of the objects of the framework. Here, the parties' interventions~$f_i$ are drawn as boxes. The directed acyclic graph~$(V,E\setminus\{(u,v)\in E \mid v\in P\})$ is hidden in the box labelled by~$\omega$.}
		\label{fig: definitions}
	\end{subfigure}
	\caption{Causal structure and the induced function.}
\end{figure}
Note that the induced function for a party~$i$ depends only on the values of the vertexes in~$\An(i)\cap P$: For party~$i$, the domain of the induced function is~$\bigtimes_{j\in P} \mathcal O_j$, yet if some~$j\in P$ is not in the ``past'' of~$i$, then the value at~$j$ has no effect on that party's induced function.
Still, we define the domain of~$\omega_i$ as the Cartesian product of all sets~$X_j$ with~$j\in P$.
This is helpful because it allows us to treat all induced functions for all parties on an equal footing, and by that, to define the function~$\omega$.
Having this, and by using the labels~$1$ to~$n$ for the parties, we can move all the vertexes and edges from the graph into a box and draw a diagram as shown in Fig.~\ref{fig: definitions}.
Now, a consistent assignment of values to the vertexes exists if and only if the function~$\omega\circ (f_1,\dots,f_n)$ has a fixed point, {\it i.e.},
	$\exists i_1,\dots,i_n: (i_1,\dots,i_n) = \omega(f_1(i_1),\dots,f_n(i_n))$.

\section{The grandfather and the information antinomy}
To present the information and the grandfather antinomy in this framework, let us first fix some notation, following that of Ref.~\cite{Baumeler2019}.
We define the objects baring no index as the collection of objects, {\it e.g.}, the inputs~$i=(i_1,\cdots,i_n)\in\mathcal{I}=\mathcal{I}_1\times\cdots\times\mathcal{I}_n$, the outputs~$o=(o_1,\cdots,o_n)\in\mathcal{O}=\mathcal{O}_1\times\cdots\times\mathcal{O}_n$, and the interventions~$f=(f_1,\dots,f_n)$.
If we wish to remove a component~$k$, then we use the notation~\mbox{$\mathcal{I}_{\setminus k}=\mathcal{I}_1\times\mathcal{I}_2\times\cdots\mathcal{I}_{k-1}\times\mathcal{I}_{k+1}\times\cdots\times\mathcal{I}_n$} etc.
We will make abuse of notation for simplicity whenever we write~$\omega(o_{\setminus k},o_k)$ or similarly.
This expression reads as $\omega(o_1,o_2,\dots,o_{k-1},o_k,o_{k+1},\dots,o_n)$.
Also, we will make use of the expression~$\omega_k(o_{\setminus k},o_k)=\omega_k(o_{\setminus k})$,
which is a short-hand expression to denote the independence of~$\omega_k$ from the argument~$o_k$:
	$\forall a,b,\in O_k,o_{\setminus k}\in O_{\setminus k}: \omega_k(o_{\setminus k},a) = \omega_k(o_{\setminus k},b)$.

The grandfather antinomy arises if there exists a choice of interventions for the parties such that no consistent assignment of values to the vertexes exist (see Fig.~\ref{fig:grandfatherexample}).
\begin{figure}[h]
	\centering
	\begin{subfigure}[t]{.4\textwidth}
		\centering
		\includegraphics[width=0.5\linewidth]{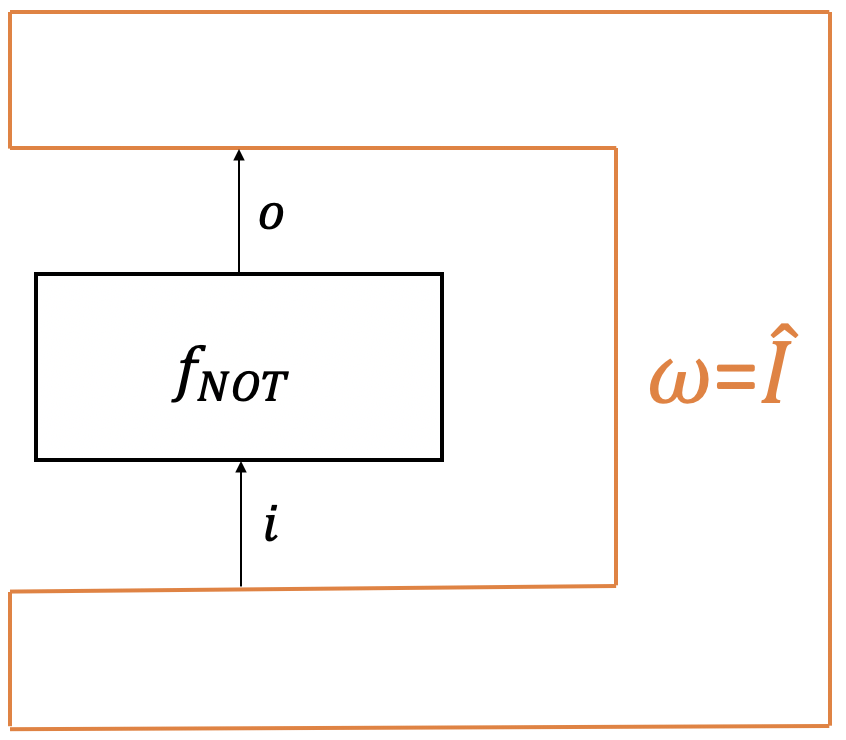}
		\caption{If the function $\omega$ is the identity and the operation of the party is to flip the inputs, then there is no fixed point (grandfather antinomy).}
		\label{fig:grandfatherexample}
	\end{subfigure}
	\qquad
	\begin{subfigure}[t]{.4\textwidth}
		\centering
		\includegraphics[width=0.5\linewidth]{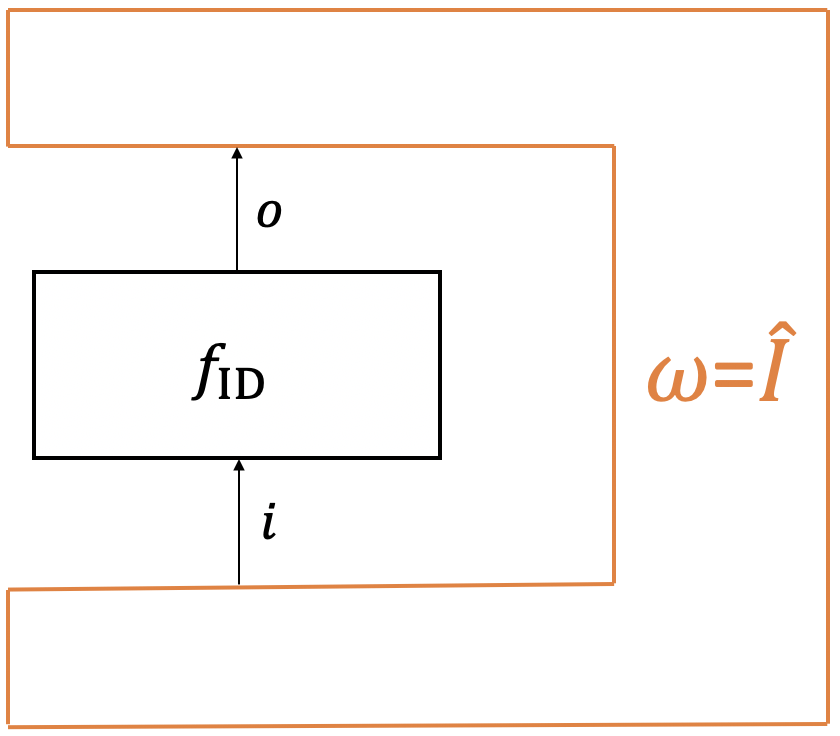}
		\caption{If both $f$ and $\omega$ are the identity channel,\\ then every possible input is a fixed point. If\\ the input variable $i$ is binary then there are two fixed points (information antinomy).}
		\label{fig:informationanti}
	\end{subfigure}
	\caption{Examples of antinomies}
	\label{fig:test}
\end{figure}
Formally, this is defined as follows.
\begin{definition}[Grandfather antinomy and process function]
	An~$n$-party induced function~$\omega$ suffers from the {\em grandfather antinomy\/} if and only if
		$\exists f: \left|\{i\mid i=\omega(f(i))\}\right|=0$.
	If an~$n$-party induced function~$\omega$ does not suffer from the grandfather antinomy, then we call~$\omega$ an~{\em $n$-party process function}.
\end{definition}
This means that in the cases where no grandfather antinomy arises, we are guaranteed to have at least one fixed point for every intervention.
The name {\em process function\/} is adequate because, as is seen at the end of this article, these functions form the set of classical and deterministic process matrices~\mbox{\cite{Oreshkov2012,Baumeler2016fixed}}.
The situation where the information antinomy arises is described by an induced function~$\omega$ where {\em multiple\/} consistent assignments of values to the vertexes exist (see Fig.~\ref{fig:informationanti}).
\begin{definition}[Information antinomy and pseudo process function]
	An~$n$-party induced function~$\omega$ suffers from the {\em information antinomy\/} if and only if
		$\exists f: \left|\{i\mid i=\omega(f(i))\}\right| \geq 2$.
	If an~$n$-party induced function~$\omega$ does not suffer from the information antinomy, then we call~$\omega$ an~{\em $n$-party pseudo process function}.
\end{definition}
When there is no information antinomy, that means that the function describing such scenarios has at most one fixed point for every choice of intervention.
It is shown in Ref.~\cite{Baumeler2019} that every process function always has a {\em unique\/} fixed point for every choice of intervention~$f$.
Here we extend that result to pseudo process functions:
\begin{theorem}\label{theorem:existenceoffixedpoint}
	Given an~$n$-party (pseudo) process function~$\omega$, there always exists a {\em unique\/} fixed point for every choice of intervention~$f=(f_1,...,f_n)$.
\end{theorem}
We prove this theorem in the next section.
A corollary of that theorem is our main message:
\begin{corollary}
	\label{corollary:equiv}
	The grandfather antinomy is equivalent to the information antinomy under intervention.
	More precisely, let~$\omega$ be an induced function, then
		$\omega$ suffers from the grandfather antinomy
		if and only if~$\omega$ suffers from the information antinomy.
\end{corollary}
This means that, if there exists an intervention~$f$ such that we can produce the grandfather antinomy in~$\omega$, then there also exists some intervention~$f'$ to produce the information antinomy, and {\em vice versa}.

\section{Properties of (pseudo) process functions}
We derive the properties of pseudo process functions and restate some results from Ref.~\cite{Baumeler2019} that hold for both, process as well as pseudo process functions.
Recently, some process functions have been characterized~\cite{Tobar2020}.
\begin{lemma}
	\label{constant1partyeqA.1}
	For an~$n$-party (pseudo) process function~$\omega$, each component~$\omega_k:\mathcal O\rightarrow \mathcal I_k$ must be constant over $\mathcal{O}_k$, {\it i.e.},
		$\forall a,b\in\mathcal{O}_k,o_{\setminus k}\in\mathcal{O}_{\setminus k}: \omega_k(o_{\setminus k},)=\omega_k(o_{\setminus k},b)$.
\end{lemma}
\begin{proof}
	Let~$\omega$ be a (pseudo) process function, and let~$\tilde o_{\setminus k}\in\mathcal O_{\setminus k}$ be some fixed input to~$\omega$ for all parties except for party~$k$.
	This allows us to define the function~$h:\mathcal O_k\rightarrow\mathcal I_k$ as~$h: x \mapsto \omega_k(x,\tilde o_{\setminus k})$.
	Showing that~$\omega$ is constant now boils down in showing that~$h$ is constant for all~$\tilde o_{\setminus k}$.
	Assume towards a contradiction that, for the given~$\tilde o_{\setminus k}$,~$h$ is not a constant, {\it i.e.}, there exist two values~$x\not=y$ such that
		$a:=h(x)\not=h(y)=:b$.
	Now, we design the interventions such that, in the case where~$\omega$ is a process function,~$\omega \circ f$ has no fixed point, and in the other case,~$\omega \circ f$ has at least two fixed points.
	In both cases, the intervention of every party~$\ell\not=k$ is~$f_\ell(z) := \tilde o_\ell$.
	This intervention is a constant function and generates the input~$\tilde o_{\setminus k}$.
	For party~$k$, in the former case ($\omega$ is assumed to be a {\em process function}), the intervention is
	\begin{align}
		f_k: z \mapsto
		\begin{cases}
			y & \text{if }z=a\\
			x & \text{otherwise.}
		\end{cases}
	\end{align}
	Indeed, in this case,~$\omega \circ f$ has {\em no\/} fixed point:
		$\forall i_{\setminus k}: \omega\circ f(a,i_{\setminus k}) = \omega(f_k(a),\tilde o_{\setminus k}) 
		= \omega(y,\tilde o_{\setminus k}) = (b,\tilde i_{\setminus k})$,
	and
		$\forall i_{\setminus k},z\not=a: \omega\circ f(z,i_{\setminus k}) = \omega(f_k(z),\tilde o_{\setminus k})
		= \omega(x,\tilde o_{\setminus k}) = (a,\tilde i_{\setminus k})$,
	where we do not need to further specify~$\tilde i_{\setminus k}$; that~$\omega \circ f$ has no fixed point is evident by looking at the~$k$-th component only.

	In the latter case ($\omega$ is assumed to be a {\em pseudo process function}), the intervention of party~$k$ is
	\begin{align}
		f_k: z \mapsto
		\begin{cases}
			x & \text{if }z=a\\
			y & \text{otherwise.}
		\end{cases}
	\end{align}
	Now, we define
	\begin{align}
		\alpha_\ell := \omega_\ell\left( f_k(a),\tilde o_{\setminus k} \right)\,,
		\quad\beta_\ell  := \omega_\ell\left( f_k(b),\tilde o_{\setminus k} \right)
		\,,
	\end{align}
	and we observe that~$(a,\alpha_{\setminus k})$ as well as~$(b,\beta_{\setminus k})$ are {\em two distinct\/} fixed points of~$\omega\circ f$:
	\begin{align}
		(a,\alpha_{\setminus k}) = \omega\left(f_k(a),f_{\setminus k}(\alpha_{\setminus k})\right)\,,
		\quad(b,\beta_{\setminus k}) = \omega\left(f_k(b),f_{\setminus k}(\beta_{\setminus k})\right)
		\,.
	\end{align}
	This holds for every~$k$ and for every~$\tilde o_{\setminus k}$.
\end{proof}

This result, being true for both process and pseudo process functions, has a clear physical interpretation: A party cannot signal back to her or himself.
Furthermore, it has an implication on single-party (pseudo) process functions:
\begin{corollary}\label{generalonepartyconstant}
	A single-party function~$\omega:\mathcal{O}\rightarrow\mathcal{I}$ is a (pseudo)
	process function 
	if and only if~$\omega$ is a constant.
\end{corollary}
\begin{proof}
	The ``only if'' case is a direct consequence of the previous lemma.
	For the ``if'' case, let~$\omega$ be a constant.
	Then it clearly has a unique fixed point from which it follows that~$\omega$ is a single-party (pseudo) process function.
\end{proof}
Note that this is a direct implication of the definitions of process function (at least one fixed point) and pseudo process function (at most one fixed point) combined with the requirement of consistency with arbitrary interventions.
If we were interested, for example, in at least or at most two fixed points, Corollary~\ref{generalonepartyconstant} would not be necessarily true.

In further discussions, we will make use of reduced functions (see Fig.~\ref{fig: reduced}):
\begin{figure}[h]
	\centering
	\includegraphics[width=.6\textwidth]{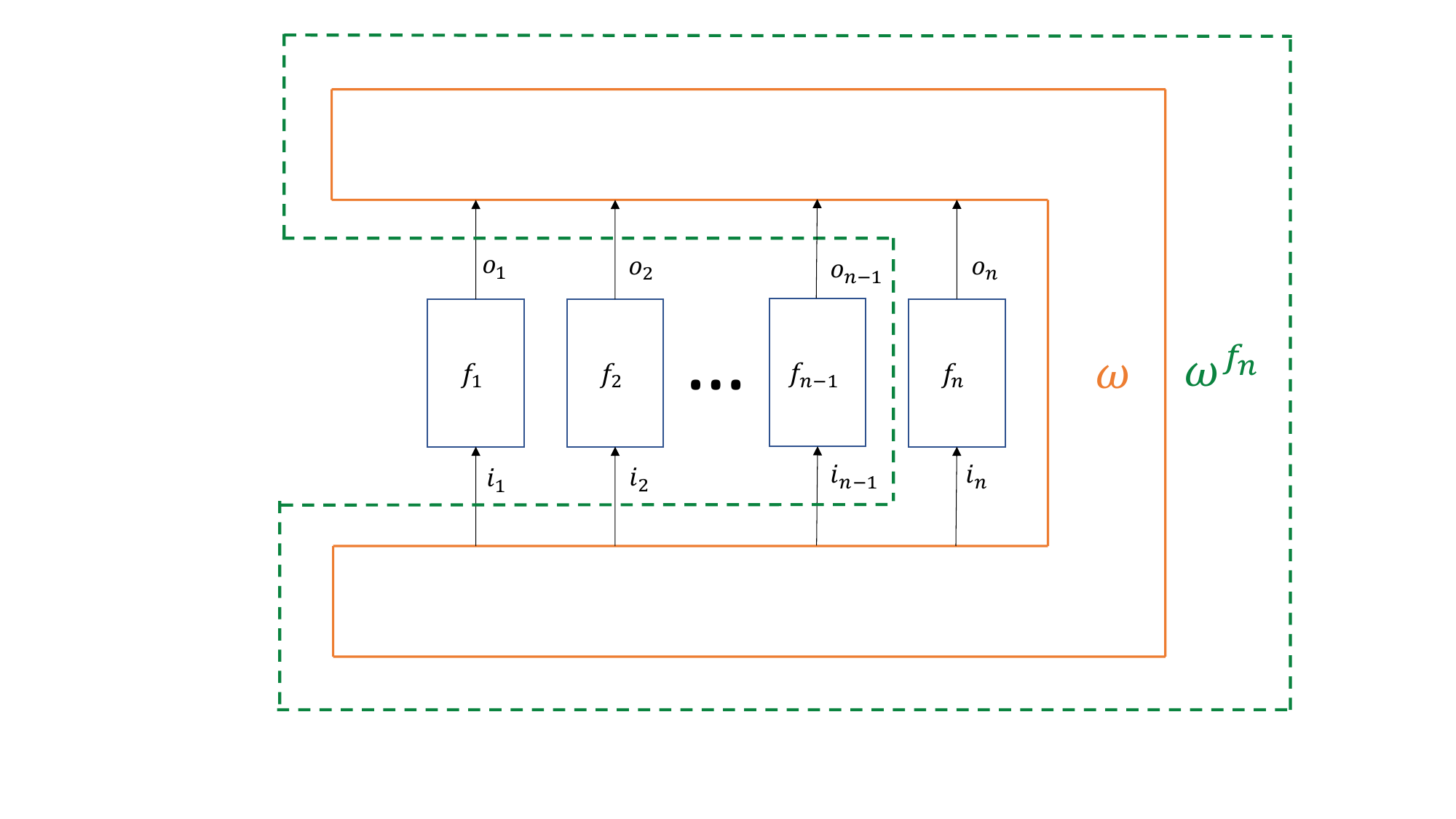}
	\caption{The reduced function $\omega^{f_n}$ is obtained by plugging in the intervention~$f_n$ for party~$n$.}
	\label{fig: reduced}
\end{figure}
\begin{definition}[Reduced function]
	Consider an~$n$-party function~$\omega:\mathcal{O}\rightarrow\mathcal{I}$, with~$O=O_1\times\dots\times O_n$ and~$I=I_1\times\dots\times I_n$, such that for every party~$k$,~$\omega_k(o)=\omega_k(o_{\setminus k})$, {\it i.e.},~$\omega_k$ is constant over~$\mathcal O_k$.
	For a fixed intervention~$f_k:\mathcal{I}_k\rightarrow\mathcal{O}_k$ of party~$k$, we define the {\em reduced function\/}~$\omega^{f_k}$ where the~$k$-th party has been ``swallowed,'' as~$\omega^{f_k}:\mathcal{O}_{\setminus k}\rightarrow\mathcal{I}_{\setminus k}$ with~$\omega^{f_k}=(\omega^{f_k}_1,\cdots,\omega^{f_k}_{k-1},\omega^{f_k}_{k+1},\cdots,\omega^{f_k}_n)$.
	Each component~$\ell\not=k$ is given by the composition of~$\omega$ with~$f_k$:
	\begin{align}
		\omega^{f_k}_\ell: \mathcal O_{\setminus k} &\rightarrow \mathcal I_\ell\\
		o_{\setminus k} &\mapsto \omega_\ell\left( o_{\setminus k}, f_k(\omega_k(o_{\setminus k})) \right)
		\,.
	\end{align}
\end{definition}

Since, according to Lemma~\ref{constant1partyeqA.1}, each component of a (pseudo) process function is constant over the same party's input, we can use the just stated definition and show that fixed points are preserved for reduced functions and {\em vice versa}:
\begin{lemma}
	\label{lemma:fppreservation}
	Let an~$(2\leq n)$-party function~$\omega:\mathcal{O}\rightarrow\mathcal{I}$ be such that for every party~$k$,~$\omega_k(o)=\omega_k(o_{\setminus k})$, then
	\begin{enumerate}
		\item if~$i\in\mathcal I$ is a fixed point of~$\omega\circ f$ for some~$f$, then, for every~$k$,~$i_{\setminus k}$ is a fixed point of $\omega^{f_k}\circ f_{\setminus k}$,
		\item if~$i_{\setminus k}\in\mathcal I_{\setminus k}$ is a fixed point of~$\omega^{f_k} \circ f_{\setminus k}$ for some~$f$ and for some~$k$, then,~$(i_k,i_{\setminus k})$ is a fixed point of~$\omega \circ f$ with~$i_k=\omega_k(f_{\setminus k}(i_{\setminus k}))$.
	\end{enumerate}
\end{lemma}
\begin{proof}
	We start with the first statement.
	The idea is to express part~$k$ of the fixed point as a function of~$\omega_k$, and then we plug it into the expression of the reduced function.
	So, we have the identity~$i_k=\omega_k(f_{\setminus k}(i_{\setminus k}))$.
	Now, for every~$\ell\not=k$, the previous definition implies
	\begin{align}
		\omega^{f_k}_\ell \circ f_{\setminus k}(i_{\setminus k}) = \omega_\ell\left( 
		f_{\setminus k}(i_{\setminus k}),
		f_k(i_k)
		\right)
		=
		\omega_\ell \circ f(i)=i_\ell
		\,.
	\end{align}

	For the second statement, we extend the fixed point~$i_{\setminus k}$ with~$\omega_k(f_{\setminus k}(i_{\setminus k}))$ for the~$k$-th component, and then show that this extended fixed point~$i$ is a fixed point of~$\omega \circ f$.
	For the~$k$-th component, this follows from the definition on how we have defined~$i_k$:~$i_k = \omega_k \circ f(i)$.
	For every component~$\ell\not=k$, we have, from the premise, that
		$i_\ell = \omega^{f_k}_\ell \circ f_{\setminus k}(i_{\setminus k})$.
	This, by definition of the reduced function, is equal to
		$\omega_\ell \circ \left( 
		f_{\setminus k}(i_{\setminus k}),
		f_k(\omega_k(f_{\setminus k}(i_{\setminus k})))
		\right)$.
	Now, by definition of~$i_k$, this expression is equal to
		$\omega_\ell \circ \left( 
		f_{\setminus k}(i_{\setminus k}),
		f_k(i_k)
		\right)
		=\omega_\ell \circ  f(i)$.
\end{proof}

We can relax that lemma to arrive at the following convenient form:
\begin{corollary}
	\label{corollary:noormanyfp}
	Let an~$(2\leq n)$-party function~$\omega:\mathcal{O}\rightarrow\mathcal{I}$ be such that for every party~$k$,~$\omega_k(o)=\omega_k(o_{\setminus k})$, then
	\begin{enumerate}
		\item if for some~$f$ the function~$\omega \circ f$ has {\em two or more\/} fixed points,
			then there exists a party~$k$ such that the function~$\omega^{f_k}\circ f_{\setminus k}$ has {\em two or more\/} fixed points.
		\item if for some~$f$ the function~$\omega \circ f$ has {\em no\/} fixed point,
			then for all parties~$k$ the function~$\omega^{f_k}\circ f_{\setminus k}$ has {\em no\/} fixed point.
	\end{enumerate}
\end{corollary}
\begin{proof}
	For the first statement, let~$i,i'$ be two fixed points of~$\omega\circ f$ that differ at position~$\ell$.
	The proof is concluded by using the first part of the previous lemma and by choosing~$k\not=\ell$.

	The second statement follows from the contrapositive of the second part of the previous lemma.
\end{proof}

Having the previous lemma at hand, we prove that the property of a function in being a (pseudo) process function transfers to less parties.
\begin{theorem}
	\label{thm:transitivity}
	If~$\omega$ is an~$(2\leq n)$-party process function, then~$\forall k,f_k:\omega^{f_k}$ is an~$n-1$-party process function.
	The same holds for pseudo process functions.
\end{theorem}
\begin{proof}
	Let~$\omega$ be a process function.
	This means that~$\forall f:\omega\circ f$ has {\em at least one\/} fixed point.
	By using the first part of Lemma~\ref{lemma:fppreservation}, we get that for all~$k,f:\omega^{f_k}\circ f_{\setminus k}$ has at least one fixed point as well.

	Let~$\omega$ be a pseudo process function.
	In this case, we have that~$\forall f:\omega\circ f$ has {\em at most one\/} fixed point.
	Assume towards a contradiction that there exists some~$k$ and some~$f$, such that~$\omega^{f_k}\circ f_{\setminus k}$ has two or more fixed points.
	Then, by the second part of Lemma~\ref{lemma:fppreservation},~$\omega\circ f$ has two or more fixed points as well, which, by definition, cannot be the case.
\end{proof}

\subsection{Existence of unique fixed point}
We can now use the results from the previous section in order to prove Theorem~\ref{theorem:existenceoffixedpoint}.
\begin{proof}
	We prove this theorem for (pseudo) process functions via the~$n$-dependent propositions
	\begin{align}
		P[n]: &\quad\omega\text{ is an~$n$-party process function}
		\Longrightarrow
		\forall f: \omega\circ f \text{ has a unique fixed point,}\notag\\
		Q[n]: &\quad\omega\text{ is an~$n$-party pseudo process function}
		\Longrightarrow
		\forall f: \omega\circ f \text{ has a unique fixed point,}\notag
	\end{align}
	and by induction over the number of parties, {\it i.e.}, we prove~$P[n]\Longrightarrow P[n+1]$ as well as~$Q[n]\Longrightarrow Q[n+1]$.

	\textit{Process functions.}
	Towards a contradiction assume that~$P[n+1]$ is false.
	The negation of~$P[n+1]$ is
	\begin{align}
		\omega\text{ is an~$(n+1)$-party process function }\wedge
		\,\exists f:\omega \circ f\text{ has two or more fixed points.}\notag
	\end{align}
	Let~$f$ be such that~$\omega\circ f$ has two or more fixed points.
	Now, by Theorem~\ref{thm:transitivity}, for any choice of~$k$,~$\omega^{f_k}$ is an~$n$-party process function.
	Furthermore, by the first part of Corollary~$\ref{corollary:noormanyfp}$, there exists some~$k$ such that~$\omega^{f_k}\circ f_{\setminus k}$ has two or more fixed points as well:~$P[n]$ is false.
	
	\textit{Pseudo process function.}
	Again, assume that~$Q[n+1]$ is false, and let~$f$ be such that~$\omega\circ f$ has no fixed point.
	By Theorem~\ref{thm:transitivity}, we also have that for any choice of~$k$,~$\omega^{f_k}$ is an~$n$-party pseudo process function.
	Now, we use the second part of Corollary~$\ref{corollary:noormanyfp}$ and we see that for any~$k$ the function~$\omega^{f_k} \circ f_{\setminus k}$ has no fixed point:~$Q[n]$ is false as well.

	Both cases, however, stand in contrast to Corollary~\ref{generalonepartyconstant} (the base case): Single-party (pseudo) process functions have a unique fixed point.
\end{proof}

By this theorem we also observe that the set of process function (which equals the set of pseudo process functions) is the set of classical and deterministic process matrices~\cite{Oreshkov2012,Baumeler2016fixed}.
Note, that by Ref.~\cite{Baumeler2016fixed}, these (pseudo) process functions can always be embedded into reversible functions as well.
Finally, examples of cyclic causal structures with interventions are known that produce dynamics incompatible with any acyclic causal structure~\cite{Baumeler2016,Baumeler2016fixed,barrett2020cyclic}.

\section{Equivalence of grandfather and information antinomy}
We prove our main statement, namely that the grandfather and the information antinomy are equivalent under intervention (see Corollary~\ref{corollary:equiv}).
\begin{proof}
	Let our causal structure be such that no grandfather antinomy arises, and let~$\omega$ be the induced function.
	This means, by definition, that if we perform any intervention~$f$, then~$\omega \circ f$ always has {\em at least one\/} fixed point.
	By Theorem~\ref{theorem:existenceoffixedpoint}, for every~$f$, the fixed point is {\em unique}.
	This again implies the absence of the information antinomy for~$\omega$.
	The same holds in the other direction.
	Suppose the causal structure to be such that its induced function~$\omega$ does not suffer from the information antinomy.
	By definition again, this means that for every choice of intervention~$f$, the function~$\omega \circ f$ has {\em at most one\/} fixed point.
	Then, by the same theorem, it follows that~$\omega \circ f$ always has {\em one\/} fixed point: The grandfather antinomy never arises.
\end{proof}

\section{Conclusion and outlook}
Following an intervention-based approach to causality, we have shown that the grandfather antinomy is equivalent to the information antinomy---in the classical case.
These antinomies might only arise in cyclic causal structures.
Cyclic causal structures have recently become a topic of research again (see, {\it e.g.}, Refs.~\cite{Oreshkov2012,barrett2020cyclic} and related work).
Cyclic causal structures, however, have a longer history: Since it is known that general relativity allows for time travel, people have been concerned about their consequences and looked for arguments against such a behaviour.
The strongest argument is the grandfather antinomy, {\it i.e.}, dynamics that lead to a logical contradiction.
Another argument, yet often neglected, is the information antinomy, where too many consistent solutions to the dynamics exist.
Here, we have shown that both problems are the same: If it were possible to generate a logical contradiction, then it would also be possible to generate a multitude of consistent solutions, and {\em vice versa}.
By showing this equivalence, we are confronted with taking the information antinomy as seriously as the grandfather antinomy.
This also motivates to exclude the information antinomy from other models~\cite{Baumeler2018}.
However, both antinomies cannot be used as arguments to rule out cyclic causal structures~\cite{phd}.

Note that the present result does not depend on the process-matrix framework~\cite{Oreshkov2012}.
On the contrary, while the process-matrix framework surely motivates this research, our result can be understood as a derivation of the classical and deterministic limit~\cite{ClassicalNC1,Baumeler2016fixed} of the process-matrix framework.
These classical deterministic processes are precisely the (pseudo) process functions described here.

Open questions are to what extent we can maintain this equivalence for probabilistic as well as for quantum causal models.
As for probabilistic models, we could augment every vertex~$v$ in the graph by a vertex~$v'$ with zero in degree and an edge from~$v'$ to~$v$, such that the ``noise'' is transferred from~$v'$ to~$v$~\cite{barrett2019quantum,barrett2020cyclic}.
If ``fine-tuning'' were forbidden, {\it i.e.}, if the properties of an induced function being a (pseudo) process function does not change under the probabilities injected at the augmented vertexes, then we suspect the same result to hold as well.\footnote{Note that the classical and probabilistic limit of the process-matrix framework~\cite{Oreshkov2012} describes stochastic processes that cannot be expressed as a convex mixture of process functions~\cite{Baumeler2016}. These stochastic processes, however, have a severe limitation: They {\em cannot\/} be derandomized.}
The quantum case is more problematic.
The reason for this is that it is not obvious on how to define both antinomies: What does it mean that the grandfather antinomy arises in a cyclic quantum causal structure?
The approach---as we did here---via fixed points does not seem to go through: Entanglement poses a problem.
Imagine that for two distinct interventions distinct quantum fixed points exist.
What happens if the parties intervene with a ``superposition'' of both interventions?
It looks like the fixed points would get entangled with the interventions, and hence, the fixed point for that ``superposition'' of interventions cannot be described separately.
Moreover, another problem with entanglement can be illustrated in a time-traveling context: If a quantum state time travels to the past, then it might still be entangled to a system in the future---temporal entanglement enters the picture and contrasts with monogamy of entanglement~\cite{Marletto2019}.

\section*{Acknowledgements}
We thank anonymous reviewers for their helpful comments.
We acknowledge support by the Austrian Science Fund~(FWF): ZK3.
\"A.B.~is also supported by the Erwin Schr\"odinger Center for Quantum Science~\& Technology~(ESQ), and the Austrian Science Fund (FWF):~F7103.

\bibliographystyle{eptcs}
\bibliography{refs}
\end{document}